\documentclass[10pt,a4paper]{article}
\usepackage{amsmath,amsthm,amssymb,pstricks,pst-node}
\usepackage{graphics,xcolor}
\usepackage{geometry}

\parskip=\medskipamount
 1   
 1  

\newcommand{\C}{\mathbb{C}}
\newcommand{\R}{\mathbb{R}}
\newcommand{\A}{{\cal A}}
\newcommand{\B}{{\cal B}}
\newcommand{\Ro}{{\cal R}}

\newcommand{\e}{\epsilon}
\newcommand{\lag}{\mathfrak{g}}

\newcommand{\Or}{{\cal O}}
\def\fpd#1#2{\frac{\partial #1}{\partial #2}}
\newcommand{\F}{\mathbb{F}}

\newcommand{\lactie}{\phi}

\newtheorem{theorem}{Theorem}
\newtheorem{lemma}{Lemma}
\newtheorem{proposition}{Proposition}
\newtheorem{definition}{Definition}
\newtheorem{remark}{Remark}

\title{Routh reduction and the class of magnetic Lagrangian systems}
\author{B. Langerock$^{a,b,c,}$\footnote{Email: bavo.langerock@ugent.be}\ , E. Garc\'{\i}a-Tora\~{n}o Andr\'{e}s$^{a,}$\footnote{Email: eduardo.gtoranoandres@ugent.be} \ and F. Cantrijn$^{a,}$\footnote{Email: frans.cantrijn@ugent.be}
\\[1.5\parskip]$^{a}$ Department of Mathematics, Ghent University\\ Krijgslaan 281 S22, B9000 Ghent, Belgium\\[1.5\parskip]
 $^{b}$ Department of Mathematics, KU~Leuven\\ Celestijnenlaan 200B, B3001 Leuven, Belgium\\[1.5\parskip]
$^{c}$ Belgian Institute for Space Aeronomy\\ Ringlaan 3, B1180 Brussels, Belgium\\
}
\date{}
\begin{document}
\maketitle\vspace{-.8cm}
\begin{abstract}
In this paper, some new aspects related to Routh reduction of Lagrangian systems with symmetry are discussed. The main result of this paper is the introduction of a new concept of transformation that is applicable to systems obtained after Routh reduction of Lagrangian systems with symmetry, so-called magnetic Lagrangian systems. We use these transformations in order to show that, under suitable conditions, the reduction with respect to a (full) semi-direct product group is equivalent to the reduction with respect to an Abelian normal subgroup. The results in this paper are closely related to the more general theory of Routh reduction by stages.
\end{abstract}

\tableofcontents
\section{Introduction}
Originally, Routh's reduction procedure is a technique in classical mechanics applicable to Lagrangian systems for which the Lagrangian is independent of one or more coordinates, also called ignorable or cyclic coordinates (see for instance~\cite{pars}). The method consists in eliminating the ignorable coordinates using a fixed value for the corresponding conserved momenta, provided that a certain regularity condition for the Legendre transform holds.

Routh reduction, and its generalization to the case of non-Abelian symmetry groups, has gained renewed attention in recent literature. For instance,  in~\cite{jalna,marsdenrouth}  a geometric formulation of this technique was given for Lagrangians of mechanical type ($L=T-V$) that are invariant under the action of an arbitrary Lie group.  Geometric Routh reduction was subsequently extended to arbitrary invariant Lagrangian systems~\cite{adamec,mestcram}. In all these treatments an essential ingredient is that the original invariant Lagrangian satisfies a certain regularity condition with respect to the Legendre transform. The situation where this condition can be relaxed is described in~\cite{BC}. Routh reduction can be seen as the Lagrangian analogue of (pre-)symplectic reduction for Hamiltonian systems (see~\cite{presympred,sympred}). This fact was in particular exploited in developing Routh reduction by stages as a special case of symplectic reduction by stages~\cite{routhstages,MarsdenHamRed}.

The dynamical system that is obtained after performing Routh reduction is represented by  Euler-Lagrange equations, not with respect to an ordinary Lagrangian but with respect to the so-called Routhian. This is a system which is of ``magnetic Lagrangian" type \cite{routhstages}. A magnetic Lagrangian system is a system derived from a Lagrangian which does not depend on some of the velocity coordinates and is subjected to a magnetic (or gyroscopic) force term, i.e. a force term obtained from a closed 2-form. Although the Lagrangian function is singular in the sense that it is independent of some velocities, due to the presence of this magnetic force term it may happen that the resulting system will be regular in the sense that it admits a symplectic formulation and that the equations of motion are of Hamiltonian type with respect to this symplectic form. We will review some aspects of these magnetic Lagrangian systems in Section~\ref{sec:mag}. For more details we refer to~\cite{routhstages}. It is important to keep in mind that every system obtained from a Lagrangian system with symmetry after Routh reduction, is a magnetic Lagrangian system. Since ordinary (regular) Lagrangian systems are trivially of magnetic Lagrangian type, with zero magnetic force term, we could say that Routh reduction is a reduction procedure within the category of magnetic Lagrangian systems.

In Section~\ref{sec:3} we present a direct way to write down the Routh reduced equations for Lagrangian systems whose configuration space is of the form $Q=S \times G$, with $G$ a Lie group with respect to which the Lagrangian is invariant. By `direct' we mean here that we do not have to compute the curvature of a chosen connection involved in the reduction procedure. Locally any manifold $Q$ with a free and proper $G$-action can be written as the product of $G$ and $S=Q/G$, indicating that the reduced equations we find are locally valid for any nontrivial action on $Q$. Our approach differs from the one in~\cite{jalna,marsdenrouth} in that instead of working with the `mechanical connection', we use the standard zero-curvature connection. As will be seen, this significantly reduces computations. Moreover, we show that the reduced equations are tightly related to the symplectic structure of the coadjoint orbits in $\lag^*$, the dual of the Lie algebra of $G$. As an example, we briefly discuss the case of a rigid body with a rotor.

The paper then continues with the introduction of a particular type of transformation between magnetic Lagrangian systems (Section~\ref{sec:4}). These transformations are diffeomorphisms between spaces on which magnetic Lagrangian systems are defined and they can map the respective dynamical systems into each other. This is subsequently applied in Section~\ref{sec:red stages} to the Routh reduction of a Lagrangian system with configuration space a product of a manifold with a semi-direct product group $G\ltimes V$ of a Lie group $G$ and a linear space $V$. In this case, there are two `natural' ways to apply Routh reduction: reducing with respect to the full symmetry group $G\ltimes V$ or with respect to the Abelian subgroup $V$.  If the dual action of $G$ on $V^*$ is free, it follows from Routh reduction by stages that there exists a symplectic diffeomorphism  relating the symplectic structures of both reduced systems. This symplectic diffeomorphism belongs to the class of transformations between magnetic Lagrangian systems we have introduced. We finally treat the case of Elroy's beanie as an illustrative example.

\paragraph{Notations and background.} For convenience we fix here some notations and we briefly recall some definitions concerning Lie group actions and connections on fibre bundles.

First of all, a point of a tangent bundle $TM$ will generally be denoted by $v_m$, meaning that $v_m \in T_mM$ with $m \in M$. If coordinates on $M$ are given by $(x^1, \ldots, x^n)$, corresponding bundle coordinates on $TM$ are written as $(x^i,v^i)$ (for $i=1, \ldots, n$).

Given a Lie group $G$, we will denote its Lie algebra by $\lag$ and the dual of its Lie algebra by $\lag^*$. We will write $\exp$ for the exponential map $\exp:\lag \to G$. The adjoint action of $G$ on $\lag$ is denoted by $Ad$ and the coadjoint action on $\lag^*$, which is defined  as the dual of the adjoint action, by $Ad^*$. A left action of a Lie group $G$ on a manifold $M$ is denoted by $\Phi^M: G\times M\to M; (g,m)\mapsto \Phi^M(g,m):=\Phi^M_g(m)$. We will also frequently write $gm$ instead of $\Phi^M_g(m)$. A left action $\Phi^M$ of a Lie group on $M$ induces an infinitesimal action at the Lie algebra level
\[
 \lactie^M_m:\lag\to T_mM; \xi \mapsto d/d\e|_0 (\exp \e\xi)(m).
 \]
Given $\xi \in \lag$, we will often consider the corresponding fundamental vector field $\xi_M$ on $M$, defined pointwise as $\xi_M(m) = \lactie^M_m(\xi)$. 

Throughout this paper we will mostly consider free and proper actions of a Lie group $G$ on a manifold $M$. This guarantees that the quotient manifold $M/G$ can be endowed with a differentiable structure such that the projection $\pi :M\to M/G$ is a principal fiber bundle (see~\cite{duistermaat2000lie,MarsdenHamRed}). Points in the manifold $M/G$ are typically denoted $[m]_G$, or simply $[m]$ when there is no risk for confusion. In carrying out Routh reduction, we have to make use of a connection on a principal fiber bundle. A connection on a principal fiber bundle  $\pi :M\to M/G$ is $\lag$-valued 1-form $\A$ on $M$ such that the following two conditions are satisfied: \begin{enumerate} \item it is $G$-equivariant, i.e. $(\Phi^M_g)^*\A = Ad_g \circ \A$ for any $g\in G$; \item for $\xi\in\lag$ arbitrary, $\A ( \xi_M)=\xi$.\end{enumerate} The kernel of $\A$ determines a left invariant distribution on $M$ which is complementary to the vertical distribution given by $\ker T\pi$. The former is referred as the $horizontal$ distribution spanned by the given connection. It is also true that any $G$-invariant distribution which is complementary to the vertical distribution determines a connection form in the obvious way. For each $\mu \in \lag^*$, we define the 1-form $\A_\mu$ on $M$ by $$ \A_\mu (m)(v_m)  = \langle \mu, \A (m)(v_m)\rangle ,$$ with $v_m \in T_mM$.

Given two bundles fibred over the same base manifold $\pi_1: P_1\to Q$ and $\pi_2: P_2\to Q$, the fibred product is the bundle with base manifold $Q$ and with total space \[ P_1\times_Q P_2:= \left\{(p_1,p_2)\in P_1\times P_2\mid \pi_1(p_1)=\pi_2(p_2) \right\}. \]

\section{Preliminaries on magnetic Lagrangian systems}\label{sec:mag}

A magnetic Lagrangian system is a Lagrangian system whose `configuration space' is the total space of a bundle $\e:P\to Q$, and where the Lagrangian is independent of the velocities tangent to the fibres of $\e$. Additionally, the system may be subjected to a magnetic force term. More precisely, we have the following definition (see e.g. \cite{routhstages}): 

\begin{definition} A magnetic Lagrangian system consists of a triple $(\e:P\to Q, L, \B)$ where $\e:P\to Q$ is a fibre bundle, $L$ is a smooth function on the fibred product $TQ\times_Q P$ and $\B$ is a closed 2-form on $P$. We say that $P$ is the configuration manifold of the system, $L$ is the Lagrangian and $\B$ is the magnetic 2-form. \end{definition}

Points in $Q$ and $P$ are usually denoted by $q$ and $p$, respectively.
Assuming $\dim Q = n$ and $\dim P = n+k$, local coordinates on $Q$ will be denoted by $(q^1, \ldots, q^n)$ and coordinates on $P$, adapted to the fibration $\e: P\to Q$, by $(q^i,p^a)$, $i=1,\ldots,n=\dim Q$, $a=1,\ldots,k=\dim P-\dim Q$. The induced bundle coordinates on $TQ\times_Q P$ are then given by $(q^i,v^i,p^a)$, where $(q^i,v^i)$ are the coordinates of a point on $TQ$. The Lagrangian $L$ is then locally expressed as a function of $(q^i,v^i,p^a)$. In particular, we note that $L$ does not depend on the velocities in the fibre coordinates and, therefore, becomes singular when interpreted as a Lagrangian on the full tangent bundle $TP$. The 2-form $\B$, written in local coordinates, reads
\[
\B=\frac12 \B_{ij} dq^i \wedge dq^j + \B_{ia}dq^i\wedge dp^a + \frac12 \B_{ab} dp^a\wedge dp^b\,.
\]
Before proceeding, we fix some further notations.

\begin{definition}\label{def:2}
Assume a magnetic Lagrangian system $(\e:P\to Q, L,\B)$ is given.
\begin{enumerate}
\item The fibred product $TQ\times_Q P$ will be abbreviated by $T_PQ$ and a point in $T_PQ$ will be denoted by $(v_q,p)$, where $v_q\in T_qQ$ and $p\in P$ is such that $\e(p)=q$. Similarly, $T_P^*Q$ denotes the fibred product $T^*Q\times_Q P$ and $(\alpha_q,p)$ represents an arbitrary point in $T^*_PQ$, with $\alpha_q \in T^*_qQ$ and $\e(p)=q$.
\item $V\e$ denotes the distribution on $P$ of tangent vectors vertical with respect to $\e.$
\item $\hat \e : TP \to T_PQ$ is the projection fibred over $P$ that maps $v_p\in TP$ onto $(T\e(v_p),p)\in T_PQ.$
\item $\tau_1:T_PQ \to TQ$ is the projection that maps $(v_q,p )\in T_PQ$ onto $v_q\in TQ.$
\item $\tau_2: T_PQ \to P$ is the projection that maps $(v_q,p)\in T_PQ$ onto $p\in P.$
\item $\pi_1:T_P^*Q \to T^*Q$ is the projection that maps $(\alpha_q,p )\in T_P^*Q$ onto $\alpha_q\in T^*Q.$
\item $\pi_2: T_P^*Q \to P$ is the projection that maps $(\alpha_q,p)\in T_P^*Q$ onto $p\in P.$
\item The Legendre transform corresponding to $L$ is the map $\F L:T_PQ \to T_P^*Q$ sending $(v_q,p)\in T_PQ$ into $(\alpha_q,p)\in T^*_PQ$, where $\alpha_q\in T^*_qQ$ is uniquely determined by the relation $$\langle \alpha_q, w_q\rangle = \left.\frac{d}{du}\right|_{u=0} L(v_q+uw_q,p),$$ for arbitrary $w_q\in T_qQ.$
\item The function on $T_PQ$ defined by $E_L(v_q,p) =\langle \F L (v_q,p), (v_q,p)\rangle - L(v_q,p)$ is called the energy of the magnetic Lagrangian system. (Here, the contraction of an element $(\alpha_q,p)\in T^*_PQ$ with $(v_q,p)\in T_PQ$ is defined naturally as $\langle (\alpha_q,p),(v_q,p) \rangle :=\langle \alpha_q,v_q\rangle$).
\item Let $\omega_{Q}=d\theta_{Q}$ be the canonical symplectic form on $T^*Q$. By means of the Legendre transform, we can pull-back the closed 2-form $\pi^*_1\omega_{Q}+\pi^*_2\B$ on $T^*_PQ$  to a closed 2-form on $T_PQ$
    \[
    \Omega^{L,\B}:=\F L^*(\pi^*_1\omega_{Q} +\pi^*_2\B).
    \]

\end{enumerate}
\end{definition}

Let us now specify the kind of dynamical system we associate with a magnetic Lagrangian system.

\begin{definition}\label{def:el} A curve $p(t)$ in $P$ is called a solution of the magnetic Lagrangian system $(\e:P\to Q,L,\B)$ if the induced curve $\gamma(t) = (\dot{q}(t), p(t))\in T_PQ$, with $q(t)=\e(p(t))$ for all $t$, satisfies the equation $$i_{\dot \gamma(t)} \Omega^{L,\B} (\gamma(t)) = -dE_L(\gamma(t)).$$
\end{definition}

Local expressions for the 2-form $\Omega^{L,\B}$ and the 1-form $dE_L$ read:
\begin{align*} \Omega^{L,\B} &= d\left(\fpd{L}{v^i} \right)\wedge d q^i + \frac12 \B_{ij} dq^i \wedge dq^j + \B_{ia}dq^i\wedge dp^a + \frac12 \B_{ab} dp^a\wedge dp^b,\\
dE_L & =v^i d\left(\fpd{L}{v^i} \right) + \fpd{L}{v^i} dv^i - dL=v^i d\left(\fpd{L}{v^i} \right) - \fpd{L}{q^i} dq^i - \fpd{L}{p^a} dp^a.
\end{align*}
With these coordinate expressions one can readily check that a curve $p(t) = (q^i(t),p^a(t))$ in $P$ is a solution of the magnetic Lagrangian system $(\e:P\to Q,L,\B)$ iff it satisfies the following set of mixed second and first order ordinary differential equations
\begin{align*} \frac{d}{dt}\left(\fpd{L}{\dot{q}^i}\right)&-\fpd{L}{q^i}=\B_{ij} \dot q^j + \B_{ia} \dot p^a\,, \\ & -\fpd{L}{p^a}=-\B_{ia} \dot q^i+ \B_{ab}\dot p^b\,,
\end{align*}
for $i=1,\ldots,n$ and $a=1,\ldots,k$. Remark that these equations are the standard Euler-Lagrange equations for the Lagrangian $\hat \e^* L$ on $TP$ subjected to a magnetic force term. Moreover, if $P=Q$ and if $\e$ is the identity, the above definition includes the standard definition of a Lagrangian system subjected to a magnetic force term. In this sense, the concept of a magnetic Lagrangian systems extends the standard concept of Lagrangian systems. For the time being we will primarily be interested in the type of magnetic Lagrangian systems called hyperregular.

\begin{definition}\label{def:hyperreg} A magnetic Lagrangian system is called regular if the following two conditions are satisfied:
\begin{enumerate} \item The 2-form $\pi^*_1\omega_{Q} +\pi^*_2\B$ is symplectic; \item $\F L$ is a local diffeomorphism.
\end{enumerate}
If, in addition, $\F L$ is a global diffeomorphism, the magnetic Lagrangian system is called hyperregular. \end{definition}

From the local expression for the magnetic 2-form, the first of these conditions is equivalent to $\det \B_{ab}\neq0$ provided $\dim P>\dim Q$. A large supply of regular magnetic Lagrangians is provided by the kind of magnetic Lagrangians which are inspired upon mechanical systems.

\begin{definition} A magnetic Lagrangian system $(\e:P\to Q,L,\B)$ is said to be of mechanical type if $L(v_q,p)= \frac12 \ll (v_q,p) ,(v_q,p)\gg_{\tau_2}- V(p)$ where $\ll \cdot , \cdot \gg_{\tau_2}$ is a metric on the vector bundle $\tau_2:T_PQ \to P$ and $V$ is a function on $P$.
\end{definition}

For simplicity, we will assume from now on that all magnetic Lagrangians systems we consider are hyperregular. The following proposition is a straightforward consequence of Definition~\ref{def:hyperreg}.

\begin{proposition} If a magnetic Lagrangian system $(\e:P\to Q,L,\B)$ is hyperregular, the 2-form $\Omega^{L,\B} = \F L^*(\pi^*_1\omega_{Q} +\pi^*_2\B)$ determines a symplectic structure on $T_PQ$.
\end{proposition}
We conclude that a hyperregular magnetic Lagrangian system induces a symplectic structure on $T_PQ$ and its dynamics is represented by the Hamiltonian vector field with respect to this symplectic structure and with the energy function as Hamiltonian:
\[
i_{X_{E_L}}\Omega^{L,\B} = - dE_L\,.
\]
Clearly, each integral curve $\gamma(t)$ of $X_{E_L}$ projects onto a solution $p(t)$ of the magnetic Lagrangian system.

{\bf Routh Reduction.} We now recall how the concept of magnetic Lagrangian systems enters the picture when dealing with Routh reduction of Lagrangian systems with symmetry. Our treatment thereby follows the symplectic reduction point of view: we introduce the symplectic structure on the tangent bundle $TQ$ where the Lagrangian is defined, and use invariance of $L$ to get a reduced system by applying the symplectic reduction procedure (for a detailed account, see~\cite{BC,sympred}).

Recall that a Lagrangian system is a pair $(Q,L)$ where $Q$ is the configuration manifold and $L$ is a smooth function on $TQ$.  Given a hyperregular Lagrangian system $(Q,L)$ (i.e. $\F L$ is a global diffeomorphism) one can define a
symplectic structure on $TQ$ by using $\F L$ to pull back the canonical symplectic form $\omega_Q$ on $T^*Q$. We denote it by $\Omega_L$, i.e. $\Omega_L=(\F L)^*\omega_Q$, and it is usually called the Poincar\'e-Cartan 2-form associated with $L$.

When a free and proper $G$-action $\Phi^Q$ on the configuration manifold $Q$ is given, a Lagrangian system $(Q,L)$ is said to be $G$-invariant if $L$ is an invariant function with respect to the lifted action of $\Phi^Q$ to $TQ$. The momentum map associated with this action is the map $J_L : TQ\to \lag^*$ defined as follows: for arbitrary  $v_q\in TQ$ and $\xi \in \lag$,
$$\langle J_L(v_q),\xi\rangle = \langle \F L(v_q) , \xi_{Q}(q)\rangle.$$
It is known that $J_L$ is equivariant in the sense that $J_L(T\Phi^Q_g(v_q)) = Ad^*_{g^{-1}}(J_L(v_q))$, where $Ad^*$ is the coadjoint action of $G$ on $\lag^*$.
We now introduce a specific regularity condition on the Lagrangian which will play an important role in the Routh reduction procedure.

\begin{definition} A $G$-invariant Lagrangian $L$ is called \emph{$G$-regular} if for every fixed $v_q \in TQ$ the map $\lag \to \lag^*, \xi \mapsto J_L(v_q+\xi_Q(q))$ is a diffeomorphism.
\end{definition}

Fix a regular value $\mu\in\lag^*$ of the momentum map $J_L$, and consider the submanifold $J_L^{-1}(\mu)$ of $TQ$ with the natural embedding $i_\mu:J_L^{-1}(\mu)\to TQ$. By equivariance of the momentum map $J_L$ it follows directly that $J_L^{-1}(\mu)$ is $G_\mu$-invariant (where $G_\mu$ is the isotropy subgroup of $\mu$). Moreover, the restricted action of $G_\mu$ on $J_L^{-1}(\mu)$ is free and proper, which provides a smooth manifold structure on the quotient $J_L^{-1}(\mu)/G_\mu$.  Due to the $G_{\mu}$-invariance of the map $J^{-1}_L(\mu) \to T_{Q/G_{\mu}}(Q/G), v_q \mapsto (T\pi(i_\mu(v_q)), [q]_{G_\mu})$ it induces a map $\Pi_\mu: J^{-1}_L(\mu)/G_\mu \to T_{Q/G_\mu}(Q/G)$. The following result immediately follows from a more general statement in \cite{routhstages}~(Proposition 7).

\begin{lemma}\label{lem:identi}
Let $L$ be a $G$-invariant and $G$-regular Lagrangian. Then $\Pi_\mu:J_L^{-1}(\mu)/G_\mu\to T_{Q/G_\mu}(Q/G)$ is a diffeomorphism.
\end{lemma}

Assume now we have chosen a principal connection $\A$ on the bundle $\pi: Q \to Q/G$ and, for a given $\mu \in \lag^*$, let ${\cal A}_\mu$ be the corresponding 1-form on $Q$. Then, the 2-form $d\A_\mu$ is projectable to a 2-form $\B_\mu$ on $Q/G_\mu$, for it is $G_\mu$-invariant and it vanishes on vector fields tangent to the fibres of the bundle $Q \to Q/G_{\mu}$. We introduce the following projection maps, with notations which are more or less in agreement with those from Definition~\ref{def:2}:
\[ \overline{\pi}_1:T^*_{Q/G_\mu}(Q/G)\to T^*(Q/G)\,, \qquad \overline{\tau}_2:T_{Q/G_\mu}(Q/G)\to Q/G_\mu \,.
\]
Moreover, we write $\omega_{Q/G}$ for the canonical 2-form on $T^*(Q/G)$.

Consider again the $G$-invariant and $G$-regular Lagrangian $L : TQ \rightarrow \mathbb{R}$ and let $\mu \in \mathfrak{g}^\ast$ be a fixed regular value of the momentum map $J_L$. Define the function $L-\hat{\A}_\mu$, with $\hat{\A}_{\mu}: TQ \to \mathbb{R}, v_q \mapsto  \A_\mu (q)(v_q )$.  This function is
$G_\mu$-invariant and, in particular, its restriction to $J^{-1}_L(\mu)$ is reducible to a function on the quotient $J^{-1}_L(\mu)/G_\mu$. Using the diffeomorphism $\Pi_{\mu}$, we can consider the corresponding function on $T_{Q/G_\mu}(Q/G)$: this the so-called Routhian, denoted as $\Ro^{\mu}$. Hence, denoting the projection $J^{-1}_L(\mu) \to J^{-1}_L(\mu)/G_{\mu}$ by $\pi_{\mu}$, we have
\[
(L - \hat{\A}_{\mu})|_{J_L^{-1}(\mu)} = (\Pi_{\mu} \circ \pi_{\mu})^*\Ro^{\mu}\,.
\]
Then we state the following result (cf. \cite{BC}, Theorem 7):

\begin{proposition}[Routh reduction] Let $L$ be a hyperregular $G$-invariant, $G$-regular Lagrangian with configuration space $Q$, and let $\mu \in \lag^*$ denote a regular value of the momentum map $J_L$. Then, the magnetic Lagrangian system $(Q/G_\mu\to Q/G, \Ro^{\mu},\B_\mu)$, as constructed above, has the property that every solution of the original Euler-Lagrange equations corresponding to the momentum value $\mu$ projects onto a solution of $(Q/G_\mu\to Q/G, \Ro^{\mu},\B_\mu)$. Conversely, every solution in $Q/G_\mu$ of  $(Q/G_\mu\to Q/G, \Ro^{\mu},\B_\mu)$ is the projection of a solution to the Euler-Lagrange equations for $L$ with momentum $\mu$.
\end{proposition}

Thus, systems of magnetic Lagrangian type appear naturally when reducing dynamics of invariant Lagrangian systems according to the Routh procedure.

\begin{remark}\textnormal{As mentioned above, throughout this paper we only consider magnetic Lagrangian systems that are hyperregular. This is the standard case, where one has Routhian reduction as a special instance of symplectic reduction. However, most of the results here can be adapted to the presymplectic setting by relating Routh reduction to presymplectic, instead of symplectic,  reduction (see~\cite{presympred}).}\end{remark}

\section{Routh reduction on product manifolds}\label{sec:3}
\subsection{General theory}
In this section we describe Routh reduction for Lagrangian systems whose configuration manifold is of the form $Q=S\times G$ and the Lagrangian $L$ is defined on $TQ= TS\times TG$. There is a left action of $G$ on $Q$ given by $\Phi^Q_{g'}(s, g)= (s,L_{g'}g)=(s,g'g)$, with $L_{g'}$ left multiplication on $G$ by $g'$. The lifted action $\Phi^{TQ}$ on $TQ$ has the form $\Phi^{TQ}: G \times TQ; (g',(v_s,v_g)) \mapsto (v_s,g'v_g)$,  where $g'v_g$ is a shorthand notation for $TL_{g'}(v_g)$. (Similarly, we will write $v_g g'$ for $TR_{g'}(v_g)$, with $R_{g'}$ right translation).

\paragraph{The left identification.} We use the left identification of $TG$ with $G\times \lag$, i.e. $v_g \mapsto (g,\xi)$ with $\xi=g^{-1}v_g\in\lag$. The tangent bundle $TQ=TS\times TG$ is then isomorphic with $TS \times G\times \lag$. The lifted action $\Phi^{TQ}$ on $TQ$ corresponds to left multiplication in the middle factor of $TS \times G\times \lag$: if $(v_s,v_g) \mapsto (v_s,g,\xi)$ then  $\Phi^{TQ}_{g'}(v_s,v_g) \mapsto (v_s,g'g,\xi)$. With this left identification in mind, the fundamental vector field $\xi_Q$ corresponding to $\xi\in\lag$ takes on the form $\xi_Q(s,g)=(0_s,g,Ad_{g^{-1}}\xi)\in TS \times G \times \lag$.

If the given Lagrangian $L$ is invariant with respect to the lifted action $\Phi^{TQ}$, i.e. $L(v_s,g\xi) = L(v_s,g'\xi)$ for any $g,g'\in G$, the corresponding expression for $L$ on $TS \times G\times \lag$ becomes independent of $G$. An invariant function $L$ determines a function $\ell$ on $TS \times \lag$:
\[
\ell(v_s,\xi):=L(v_s,g\xi).
\]
The purpose now is to express Routh reduction of the $G$-invariant Lagrangian system $(Q=S \times G, L)$ in terms of the function $\ell$. It turns out that in this case we can write down an explicit form for the reduced equations of motion.

\paragraph{The momentum map.} We recall the definition of the momentum map $J_L$, evaluated at $(v_s,g\xi) \in TQ = TS \times TG$, and we substitute $\ell(v_s,\xi)$ for $L(v_s,g\xi)$:
\begin{align*}
\langle J_L(v_s,g\xi),\eta\rangle &= \left.\frac{d}{d\e}\right|_{\e=0} L\big( v_s,g(\xi +\e Ad_{g^{-1}}\eta)\big) = \left.\frac{d}{d\e} \right|_{\e=0}\ell(v_s,\xi + \e Ad_{g^{-1}} \eta) \\ &= \langle\F_2 \ell(v_s,\xi),Ad_{g^{-1}} \eta\rangle\,,
\end{align*}
for all $\eta \in \lag$, and where $\F_2\ell: TS \times \lag  \to \lag^*$ is defined by the relation
\[
\langle\F_2 \ell(v_s,\xi), \tau\rangle = \left.\frac{d}{d\e}\right|_{\e=0}\ell(v_s,\xi + \e \tau), \;\;\;\; \mbox{for all}\; \tau \in \lag\,.
\]
Consequently, we conclude from the above that \[ J_L(v_s,g\xi) = Ad^*_{g^{-1}}\F_2\ell(v_s,\xi)\,. \] For every $\mu\in\lag^*$ the equation $J_L(v_s,g\xi)=\mu$ can be equivalently written as $\F_2\ell(v_s,\xi)=Ad^*_g\mu$.

\paragraph{$G$-regularity.} Next, we investigate the $G$-regularity of $L$. Recall that this is in fact a condition on $J_L$: for every $g\in G,\xi\in\lag, v_s\in TS$ and $\nu\in\lag^*$, there exists a unique $\eta\in\lag$ such that \[ J_L(v_s,g\xi + g\eta)= \nu. \] Using the foregoing result this is equivalent to $\F_2\ell(v_s,\xi+ \eta) = Ad_{g}^*\nu$. In particular, it follows that $G$-regularity of $L$ is equivalent here to the condition that the map $\F_2\ell(v_s,\cdot):\lag\to\lag^*$ is invertible. Denote the inverse map by $\chi^{(v_s)}: \lag^* \to \lag$, where $\chi^{(v_s)}$ depends smoothly on the `parameter' $v_s \in TS$. For later use we define $\chi^{(v_s)}_\mu$ to be the restriction of $\chi^{(v_s)}$ to the co-adjoint orbit $\mathcal{O}_{\mu}$ of some $\mu\in\lag^*$.

\paragraph{The connection 1-form.} The standard connection on $Q = S \times G$ regarded as a $G$-bundle over $S$ is $\A(v_s,v_g)= v_gg^{-1} \in \lag$. This is in fact the trivial extension to $S \times G$ of the canonical connection on $G$ associated to the left multiplication. The corresponding map from $TS \times G\times \lag \to \lag$  is $(v_s,g\xi) \mapsto \A(v_s,g\xi)=Ad_g\xi$. For $\A$ to be a principal connection it should satisfy two conditions: (1) when contracted with a fundamental vector field it should provide the corresponding Lie algebra element, and (2) it should be equivariant. Using again the left identification, these conditions are easily verified: \begin{enumerate} \item $\A(\xi_Q(s,g))=\xi$ for $\xi \in\lag$ arbitrary. Indeed, from $\xi_Q(s,g)=(0_s,gAd_{g^{-1}} \xi)$, it follows that $\A(0_s,gAd_{g^{-1}}\xi) = Ad_g( Ad_{g^{-1}}\xi) = \xi$. \item $\A$ is equivariant if $\A(v_s,g'g\xi)  = Ad_{g'} \A(g\xi)$ for all $g',g\in G$ and $\xi\in\lag$. The left-hand side of this equality becomes $\A(v_s,g'g\xi)= Ad_{g'g} \xi = Ad_{g'} Ad_g\xi$, which is precisely equal to the right-hand side. \end{enumerate} The horizontal distribution for this connection is the subbundle $\ker\A = TS\times 0_G$ of $TQ$.

Next, for $\mu \in \lag^*$ we consider the 1-form $\A_{\mu}$ on $S \times G$. Evaluating it on a horizontal tangent vector $(v_s,0_g)$ and on a vertical tangent vector $(0_s,g\xi)$ at a point $(s,g)$ we find, respectively, $\A_{\mu}(v_s,0_g) = 0$ and $\A_{\mu}(0_s,g\xi) = \langle \mu, Ad_g\xi \rangle$. Now, we can compute its exterior derivative $d\A_{\mu}$ by evaluating it on two horizontal vector fields, on two vertical vector fields and on a horizontal and vertical vector field, respectively. Using Cartan's formula for the exterior derivative and taking into account the $G$-equivariance of the connection, one obtains after a routine calculation that
\[
d\A_{\mu}(s,g)\big((v_s,g\xi),(w_s,g\xi')\big) = \langle Ad^*_g\mu, [\xi,\xi'] \rangle\,.
\]

\paragraph{Computation of the 2-form $\B_\mu$.} We now reduce the 2-form $d\A_{\mu}$ to $Q/G_{\mu}= S \times G/G_{\mu}$. Recall that $G/G_\mu$ is diffeomorphic to $\Or_\mu$, with diffeomorphism defined by $G/G_{\mu} \to \Or_\mu, [g]_{G_\mu} = G_\mu g \to Ad^*_g\mu$. The tangent space to $\Or_{\mu}$ at some $\nu \in \Or_{\mu}$ is given by (cf.\ \cite[section 14.2]{marsden1999introduction})
\[
T_{\nu}\Or_{\mu} = \{ad_{\xi}^*\nu\,|\,\xi \in \lag\}\,.
\]
Define the 2-form $\B_\mu$ on $S \times \Or_{\mu}$ by the following prescription:
\[
\B_{\mu}(s,\nu)\left((v_s,ad_{\xi}^*\nu),(w_s,ad_{\xi'}^*\nu)\right) = \langle \nu,[\xi,\xi']\rangle\,.
\]
One can verify that this is a closed 2-form. In fact, $\B_\mu$ is the (trivial) extension to $S \times \Or_\mu$ of the standard Kostant-Kirillov-Souriau symplectic form $\omega^+$ on $\Or_\mu$ (see \cite[section 14.2]{marsden1999introduction}). Using the above formula for $d\A_\mu$, one easily verifies that $\B_\mu$ is the projection of $d\A_{\mu}$ on $S \times \Or_\mu$.

\paragraph{Routh reduction.} Since in the case under consideration we have that $J^{-1}_L(\mu)/G_\mu\cong TS\times \Or_\mu \cong T_{S \times \Or_\mu}S$ and, hence, the map $\Pi_\mu$ from Lemma~\ref{lem:identi} becomes the identity, we find the Routhian $\Ro^{\mu}$ by taking the restriction of $\ell- \hat{\A }_{\mu}$ to $J_L^{-1}(\mu)$ and projecting it onto $J^{-1}_L(\mu)/G_\mu$. The expression for $\Ro^{\mu}$ therefore reads
\[
\Ro^{\mu}(v_s,\nu)= \left(\ell(v_s,\xi) - \langle \nu, \xi\rangle\right)|_{\xi=\chi^{(v_s)}_\mu(\nu)}\,,
\]
where $\nu = Ad^*_g\mu \in \Or_\mu$.

\paragraph{The local equations of motion.} The equations of motion for the magnetic Lagrangian system $(S \times \Or_\mu \rightarrow S, \Ro^{\mu}, \B_\mu)$ can be split into two parts: the part describing the evolution on $TS$ and the part describing the evolution on $\Or_\mu$. We start with the latter.

Let $e^a$ denote a basis for $\lag^*$ and let $\dot \nu' = \dot \nu'_ae^a$ be an arbitrary tangent vector to $\Or_\mu$: \begin{align*}
-\langle dE_{\Ro^\mu}(v_s,\nu),(0_{v_s},\dot \nu')\rangle & =\fpd{\Ro^{\mu}}{\nu_a} (v_s,\nu)\dot \nu'_a \\
&= \left\langle \F_2 \ell(v_s,\chi^{(v_s)}(\nu)), \fpd{\chi^{(v_s)}}{\nu_a}(\nu)\dot\nu'_a \right\rangle \\
&\quad  - \dot\nu'_a \big(\chi^{(v_s)}(\nu)\big)^a - \left\langle \nu , \fpd{\chi^{(v_s)}}{\nu_a}(\nu)\dot \nu'_a \right\rangle \\
& =   -\langle \dot \nu',\chi^{(v_s)}(\nu)\rangle.
\end{align*}

Therefore, the reduced equation of motion is \[\B_\mu(\nu) (\dot\nu,\dot\nu')= -\langle \dot \nu',\chi^{(v_s)}_\mu(\nu)\rangle\] with $\dot \nu'$ arbitrary in $T\Or_\mu$. We conclude that one component of the Euler-Lagrange equation is precisely  $\dot \nu = ad^*_{\chi^{(v_s)}_\mu(\nu)} \nu$. The component in $T^*S$ has the structure of standard Euler-Lagrange equations: if $(x^i)$ is a coordinate system on $S$, then the equations of motion are
\[
\frac{d}{dt}\left( \fpd{\Ro^{\mu}}{\dot x^i}(x^i,\dot x^i,\nu)\right) - \fpd{\Ro^{\mu}}{x^i}(x^i,\dot x^i,\nu)=0.
\]

Summarizing, we have proved the following result:
\begin{theorem}\label{thm:left} Let $\ell$ denote the restriction to $TS \times \lag$ of a left $G$-invariant Lagrangian $L$ on $T(S \times G)$ and let $\F_2\ell: TS \times \lag \to \lag^*$ denote the fibre derivative w.r.t. the second argument. Fix an element $\mu$ in $\lag^*$ and assume that there exists a map $\chi^{(v_s)}:  \lag^*\to \lag$ which smoothly depends on $v_s \in TS$, such that $\F_2\ell(v_s,\chi^{(v_s)}(\nu))\equiv \nu$ for arbitrary $(v_s,\nu)\in TS \times \lag^*$.  Then, the reduced  system is the magnetic Lagrangian system $(S \times \Or_{\mu} \to S, \Ro^{\mu}, \B_{\mu})$ where the 2-form $\B_{\mu}$ on $S \times \Or_{\mu}$ and the Routhian $\Ro^{\mu}$ on $TS \times \Or_{\mu}$ are given by, respectively,
\[
\B_{\mu}(s,\nu)\left((v_s,ad_{\xi}^*\nu),(w_s,ad_{\xi'}^*\nu)\right) = \langle \nu,[\xi,\xi']\rangle\,,
\]
and
\[
\Ro^{\mu}(v_s,\nu) = \big(\ell(v_s,\xi)- \langle \nu, \xi\rangle\big)_{\xi=\chi^{(v_s)}_\mu(\nu)}.
\]
(Here, $\chi^{(v_s)}_\mu$ is the restriction of $\chi^{(v_s)}$ to the coadjoint orbit $\Or_\mu$). In a local coordinate chart $(x^i)$ on $S$, the equations of motion for the reduced system are a system of coupled first and second order differential equations: \begin{equation}\label{eq:thmredright}
\left\{\begin{array}{l}\displaystyle \dot \nu = ad^*_{\chi^{(v_s)}_\mu(\nu)} \nu, \\[.2cm] \displaystyle \frac{d}{dt}\left( \fpd{\Ro^{\mu}}{\dot x^i}(x^i,\dot x^i,\nu)\right) - \fpd{\Ro^{\mu}}{x^i}(x^i,\dot x^i,\nu)=0.
\end{array}\right.
\end{equation}
\end{theorem}

A similar result holds in case we are dealing with a Lagrangian on $Q = S \times G$ which is right invariant, i.e. which is invariant under the lifted action of $\Psi^Q: G\times Q \to Q$, $(g',(s,g))\mapsto (s,gg')$. Given the appropriate function $\ell$, the reduced equation of motion in this case  will slightly differ from those obtained above: the component along $\Or_ {\mu}$ becomes $\dot \nu = -ad^*_{\chi^{(v_s)}_\mu(\nu)} \nu$.

\subsection{Example: The rigid body with a rotor} We consider a rigid body with a single rotor along the third principal axis of the body. This example is taken from~\cite{satellite}. The configuration space of this system is $Q=S^1 \times SO(3)$, where $SO(3)$ is the configuration space of the rigid body and $S^1$ measures the angle of the rotor relative to the body frame which we denote by $x$. In the body frame of the principal inertia axes, the (reduced) Lagrangian $\ell: TS^1 \times so(3) \rightarrow \R$ has the following expression:
\[
\ell(x,\dot x, \omega) = \frac12 \left(\omega \mathbb{I} \omega + (\omega+ \alpha ) \mathbb{J} (\omega+\alpha)\right), \]
where $\mathbb{I}$ and $\mathbb{J}$ are the inertia tensors corresponding to the rigid body and the rotor, respectively, $\omega=(\omega_1,\omega_2,\omega_3)$ denotes the angular velocity of the body and  $\alpha:=(0,0,\dot x)$ corresponds to the angular velocity of the rotor, both in the body frame. Introducing the quantities $\lambda_i=I_i+J_i$, $i=1,2,3$, the Lagrangian becomes explicitly
\[
\ell(x, \dot x, \omega) =  \frac12 (\lambda_1\omega_1^2 + \lambda_2\omega_2^2 + \lambda_3\omega_3^2 + J_3 \dot x^2) + J_3 \dot x \omega_3.
\]
The map $\mathbb{F}_{2}\ell(x,\dot x, \cdot):\mathbb{R}^3(\cong so(3))\to\mathbb{R}^3(\cong so^*(3))$ is given by
\[
\mathbb{F}_2 \ell (\dot x,\omega) = (\lambda_1 \omega_1,\lambda_2\omega_2,\lambda_3\omega_3+J_3 \dot x),
\]
and its inverse equals
\[
\chi^{(x,\dot{x})}(m)= \left(\frac{1}{\lambda_1}m_1,\frac{1}{\lambda_2}m_2,\frac{1}{\lambda_3}(m_3 - J_3\dot x)\right), \]
where, in the notations of the previous section, $m = (m_1,m_2,m_3)\in \R^3\cong so^*(3)$ corresponds to $\nu$ and $(x,\dot x)$ corresponds to $v_s$. For the actual computation of the Routhian, we use the property that for Lagrangians of mechanical type with potential energy $V(s)$, the Routhian can be computed from the following identity~\cite{pars}: \[
2\big(\Ro^\mu(v_s,\nu) + V(s)\big)=\left.\left(\langle \F_1\ell(v_s,\xi) , v_s\rangle - \langle \F_2\ell(v_s,\xi), \xi \rangle \right)\right|_{\xi=\chi^{(v_s)}_\mu(v_s,\nu)}.
\]
Using the above expression, the Routhian is obtained in a straightforward way:
\begin{align*}
\Ro^{m_0}(x,\dot x, m) & = \frac{1}{2}\left((J_3\dot x^2+J_3\dot x\omega_3)-\lambda_1\omega_1^2-\lambda_2\omega_2^2-\omega_3(\lambda_3\omega_3+J_3\dot x)\right)\big{|}_{\omega=\chi^{(x,\dot x)}_{m_0}(m)}\\
& = \frac{1}{2}\left(J_3\dot x^2\left(1-\frac{J_3}{\lambda_3}\right) - \frac{m_1^2}{\lambda_1} - \frac{m_2^2}{\lambda_2} - \frac{m_3^2}{\lambda_3}\right)+\frac{J_3}{\lambda_3}\dot x m_3\\
& = \frac{1}{2}\left(\frac{J_3 I_3}{\lambda_3}\dot x^2 - \frac{m_1^2}{\lambda_1} - \frac{m_2^2}{\lambda_2} - \frac{m_3^2}{\lambda_3}\right)+\frac{J_3}{\lambda_3}\dot x m_3 .
\end{align*}
Note that the difference with the Routhian obtained in~\cite{jalna}, which was computed using the mechanical connection, is the appearance of the product term $\dot x m_3$. The reduced equations on $so^*(3)$ read $\dot m = ad^*_{\chi^{(x,\dot x)}_{m_0}(m)} m = m\times\chi^{(x,\dot x)}_{m_0}(m)$. Finally, the full reduced set of equations of motion corresponding to $\Ro^{m_0}$ read:
\begin{align*}
\dot{m}_{1} & =  \left(\frac{1}{\lambda_{3}}-\frac{1}{\lambda_{2}}\right)m_{2}m_{3}-\frac{m_{2}J_{3}}{\lambda_{3}} \dot x, & \dot{m}_{2} & =  \left(\frac{1}{\lambda_{1}}-\frac{1}{\lambda_{3}}\right)m_{1}m_{3}+\frac{m_{1}J_{3}}{\lambda_{3}} \dot x,\\ \dot{m}_{3} & =  \left(\frac{1}{\lambda_{2}}-\frac{1}{\lambda_{1}}\right)m_{1}m_{2}, & I_3\ddot x  &= -\dot m_3 .
\end{align*}

\begin{remark}\textnormal{In ~\cite{CLeuler} the previous example is also treated in the context of controlled Lagrangians (see also ~\cite{CLpotential,CLmatching,CLgyr}). It would be of interest to investigate in detail the connections between the two approaches.}
\end{remark}
\section{Transformations between magnetic Lagrangian systems}\label{sec:4}

We now introduce a particular type of transformations relating two magnetic Lagrangian systems. This can be seen as a generalization of the concept of point transformations in Lagrangian mechanics and is inspired upon the techniques encountered in Routh reduction. It allows one to transform a magnetic Lagrangian system into a new magnetic Lagrangian system with an enlarged configuration space $P$, but with a greater number of `constraints' in order to compensate for the raise in degrees of freedom.

\subsection{Compatible transformations}
\begin{definition} Let $\e^{(1)}:P_1\to Q_1$ and $\e^{(2)}:P_2\to Q_2$ be two fiber bundles. If $F:P_1\to P_2$ and $f:Q_2\to Q_1$ are two surjective submersions we say that the pair $(F,f)$ forms a \emph{transformation pair} between both bundles if the following equality holds:
\[
f \circ \e^{(2)} \circ F = \e^{(1)},
\]
and all the arrows in Figure~\ref{fig:diagram1} represent fiber bundles.
\begin{figure}[htb]
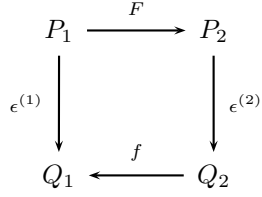
\centering
\[
\vspace{-.1cm}
\setlength{\arraycolsep}{.8cm}
\begin{array}{ccc}
\Rnode{a}{P_1} & \Rnode{b}{P_2}\\[1.5cm]
\Rnode{c}{Q_1} & \Rnode{d}{Q_2}
\end{array}
\psset{nodesep=5pt,arrows=->}
\everypsbox{\scriptstyle}
\ncLine{a}{b}\Aput{F}
\ncLine{a}{c}\Bput{\e^{(1)}}
\ncline{d}{c}\Bput{f}
\ncLine{b}{d}\Aput{\e^{(2)}}
\]
\caption{Transformation pair}\label{fig:diagram1}\end{figure}
\end{definition}

Putting $\dim Q_i = n_i$ and $\dim P_i = n_i + k_i$ for $i=1,2$,  it immediately follows from this definition that $n_1 + k_1 \geq n_2 + k_2$ and $n_1 \leq n_2$, from which we deduce that the fiber dimensions must satisfy $k_1 \geq k_2$.

A transformation pair $(F,f)$ determines a relation between the points of $T_{P_1}Q_1$ and $T_{P_2}Q_2$. \begin{definition} Given a transformation pair $(F,f)$, then two points $(v_{q_i},p_i)\in T_{P_i}Q_i$, $i=1,2$ are {\em $(F,f)$-compatible} if $F(p_1)=p_2$ and $Tf(v_{q_2})=v_{q_1}$.
\end{definition}
We denote by $Vf (\subset TQ_2)$ the subbundle of $f$-vertical tangent vectors to $Q_2$ and its dual by $V^*f$.

Suppose a transformation pair is given between the bundles $\e^{(1)}$ and $\e^{(2)}$. The projections of the fibred product $T_{P_i}Q_i$ onto $TQ_i$ and $P_i$ will be denoted by $\tau^{(i)}_1$ and $\tau^{(i)}_2$, respectively, and the projections of $T^*_{P_i}Q_i$ onto $T^*Q_i$ and $P_i$ by $\pi^{(i)}_1$ and $\pi^{(i)}_2$, respectively, for $i=1,2$.

\begin{definition} A smooth mapping $\psi: T_{P_1}Q_1 \to T_{P_2}Q_2$ is {\em compatible} with the transformation pair $(F,f)$ if the following conditions are verified: (i) $Tf \circ \tau^{(2)}_1 \circ \psi = \tau^{(1)}_1$, and (ii) $\tau^{(2)}_2 \circ \psi = F \circ \tau^{(1)}_2$.
\end{definition}

It is easily verified that a map $\psi$ is compatible with $(F,f)$ iff the image of a point and the point itself are $(F,f)$-compatible.  In particular this implies that a compatible transformation $\psi$ makes the diagrams in Figure~\ref{fig:comp} commute.
\begin{figure}[thb]
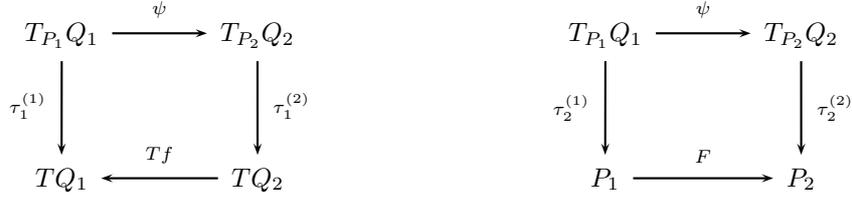

\[
\setlength{\arraycolsep}{.8cm}
\begin{array}{ccc}
\Rnode{a}{T_{P_1}Q_1} & \Rnode{b}{T_{P_2}Q_2}\\[1.5cm]
\Rnode{c}{TQ_1} & \Rnode{d}{TQ_2}
\end{array} \hspace{2cm}
\begin{array}{ccc}
\Rnode{e}{T_{P_1}Q_1} & \Rnode{f}{T_{P_2}Q_2}\\[1.5cm]
\Rnode{g}{P_1} & \Rnode{h}{P_2}
\end{array}
\psset{nodesep=5pt,arrows=->}
\everypsbox{\scriptstyle}
\ncLine{a}{b}\Aput{\psi}
\ncLine{a}{c}\Bput{\tau_1^{(1)}}
\ncLine{d}{c}\Bput{Tf}
\ncLine{b}{d}\Aput{\tau_1^{(2)}}
\ncLine{e}{f}\Aput{\psi} \ncline{f}{h}\Aput{\tau_2^{(2)}} \ncline{e}{g}\Bput{\tau_2^{(1)}} \ncline{g}{h}\Aput{F}
\]
\caption{A diagram for $(F,f)$-compatible transformations.}\label{fig:comp}
\end{figure}
It is worth deriving a coordinate expression for compatible transformations. We choose coordinate charts that are adapted to the fibrations, i.e. starting from coordinates $(q^i)$ on $Q_1$, consider coordinates $(q^i,\bar q^a)$ on $Q_2$ adapted to $f$. We denote the $\e^{(2)}$ adapted coordinates on $P_2$ by $(q^i,\bar q^a,\bar p^\alpha)$ and, finally, $(q^i,\bar q^a,\bar p^\alpha,p^\gamma)$ on $P_1$ adapted to $F$. A $(F,f)$-compatible transformation then assumes the form: \[\psi(q^i,\dot q^i, \bar q^a,\bar p^\alpha, p^\gamma) = (q^i,\bar q^a,\dot q^i,\psi^a(q^i,\dot q^i, \bar q^a,\bar p^\alpha, p^\gamma), \bar p^\alpha),\] with the only non-trivial components in velocities vertical to $f$. Remark that if a compatible transformation is a diffeomorphism, then for every $p_2\in P_2$, the fibre $F^{-1}(p_2)$ is diffeomorphic to $V_{q_2}f$, where $q_2=\e^{(2)}(p_2)$.  This assumption in particular implies that $2n_1 + k_1 = 2n_2 + k_2$, i.e.
\[
\dim T_{P_1}Q_1 = \dim T_{P_2}Q_2.
\]

\begin{remark}\textnormal{We already pointed out that an ordinary Lagrangian system is a special instance of a magnetic Lagrangian system with $P\equiv Q$ and $T_PQ\equiv TQ$. Consider a point transformation between $Q_2$ and $Q_1$, i.e., a diffeomorphism $f:Q_2\to Q_1$. Then the pair $(F=f^{-1},f)$ is a transformation pair and the tangent lift of $f^{-1}$, i.e. $\psi=Tf^{-1}$, is a compatible transformation.}
\end{remark}

\subsection{A special class of compatible transformations} We now proceed with the case where in addition to a transformation pair $(F,f)$ between $\e^{(1)}$ and $\e^{(2)}$, a hyperregular magnetic Lagrangian system  $(\e^{(2)}, L_2, \B_2)$ is given. It is our purpose to construct a class of $(F,f)$-compatible transformations by means of the Lagrangian $L_2$ such that, under suitable regularity conditions: (1) it is a diffeomorphism between $T_{P_1} Q_1$ and $T_{P_2} Q_2$, and (2) it pulls-back  the Hamiltonian vector field $X_{E_{L_2}}$ on $T_{P_2} Q_2$ to a vector field on $T_{P_1} Q_1$ which is the Hamiltonian vector field associated to a new magnetic Lagrangian system on $\e^{(1)}$. The main theorem of this section then relates the symplectic structures and the Hamiltonian dynamics associated to the magnetic Lagrangian systems on $\e^{(1)}$ and $\e^{(2)}$. The construction of the $(F,f)$-compatible transformation consists of three steps.

{\underline{Step~1}}. Consider the map $\alpha_{L_2}: T_{P_2}Q_2 \to V^*f$ which is defined as $\pi^{(2)}_1 \circ {\F L}_2: T_{P_2}Q_2 \to T^*Q_2$ composed with the projection of $T^*Q_2$ onto $V^*f$. \begin{definition} Given a transformation pair $(F,f)$ between $\e^{(1)}$ and $\e^{(2)}$ and a Lagrangian system $(\e^{(2)},L_2,\B_2)$. The Lagrangian $L_2$ is $f$-regular if for any given $(v_{q_2},p_2) \in T_{P_2}Q_2$ the map $\alpha^{(v_{q_2},p_2)}_{L_2}: V_{q_2}f \to V^*_{q_2}f, w_{q_2} \mapsto \alpha_{L_2}(v_{q_2} + w_{q_2},p_2)$ is a diffeomorphism.\end{definition}

{\underline{Step~2}}. Fix any surjective submersion $\beta: P_1 \to V^*f$ satisfying the regularity condition that $\beta|_{F^{-1}(p_2)}: F^{-1}(p_2) \to V^*_{q_2}f$ is a diffeomorphism for each $p_2 \in P_2$, with $\e^{(2)}(p_2) = q_2$. For this to hold, the dimension of the fibres of $f$ should be equal to the dimension of the fibres of $F$ which is necessary for a $(F,f)$-compatible transformation to be a diffeomorphism. 

{\underline{Step~3}}. We now define a map $\psi_{L_2,\beta}$ from $T_{P_1}Q_1$ to $T_{P_2}Q_2$ associated to a $f$-regular Lagrangian $L_2$ and a fixed $\beta$ as in the previous step. Let $(v_{q_1},p_1)$ be an arbitrary point in $T_{P_1}Q_1$  and let $(v_{q_2},p_2)\in T_{P_2}Q_2$ be compatible with it (such a point always exists).  Due to the $f$-regularity of $L_2$, there is a unique tangent vector $w_{q_2}\in V_{q_2}f$ that satisfies $\alpha^{(v_{q_2},p_2)}_{L_2}(w_{q_2}) = \beta(p_1)$, or alternatively
\[
\pi^{(2)}_1\big({\F L}_2{(v_{q_2}+w_{q_2},p_2)}\big)|_{Vf}= \beta(p_1).
\]
We take the point $(v_{q_2}+w_{q_2},p_2)$ as the image of $(v_{q_1},p_1)$ under $\psi_{L_2,\beta}$. Because $(v_{q_2},p_2)$ is compatible with $(v_{q_1},p_1)$ and because $L_2$ is $f$-regular, the construction  is independent of the choice of $v_{q_2}$. If $\beta$ satisfies the regularity condition from step 2, the map $\psi_{L_2,\beta}$ admits an inverse and is a diffeomorphism. The following proposition summarizes the above.

\begin{proposition}\label{prop:compatible} Let $L_2$ be $f$-regular and $\beta:P_1\to V^*f$ be arbitrary. Then the map $\psi_{L_2,\beta}: T_{P_1}Q_1\to T_{P_2}Q_2$ constructed above is uniquely determined from the two conditions: \begin{enumerate}
\item $\psi_{L_2,\beta}$ is a $(F,f)$-compatible transformation;
\item $\pi^{(2)}_1\big({\F L}_2{(\psi_{L_2,\beta}(v_{q_1},p_1))}\big)|_{Vf}= \beta(p_1)$ for arbitrary $(v_{q_1},p_1) \in T_{P_1}Q_1$.
\end{enumerate} If $\beta$ satisfies the regularity condition from Step~2, then $\psi_{L_2,\beta}$ is a diffeomorphism.
\end{proposition}

A $(F,f)$-compatible transformation of the form $\psi_{L_2,\beta}$ can be used to pull-back the symplectic 2-form and the energy of the magnetic Lagrangian system $(\e^{(2)},L_2,\B_2)$ to the manifold $T_{P_1}Q_1$. We show in Theorem~\ref{thm:main} that a magnetic Lagrangian system on $\e^{(1)}$ exists whose associated symplectic 2-form is precisely $\psi^*_{L_2,\beta}\Omega^{L_2,\B_2}$ and whose energy is $\psi^*_{L_2,\beta} E_{L_2}$.

For the definition of this magnetic Lagrangian system on $\e^{(1)}$, we consider a connection $\A$ on the bundle $f: Q_2 \to Q_1$. Recall that $\A$ may be represented as a $Vf$-valued 1-form on $Q_2$, satisfying $\A(v_{q_2})=v_{q_2} $, for all $v_{q_2}\in Vf$. It induces an splitting $TQ_2=Hf\oplus Vf$, where $Hf=\ker\A$ is the  horizontal subbundle of $TQ_2$ defined by $\A$. It is standard to denote the horizontal and vertical components of a tangent vector $v_{q_2}$ by $v^H_{q_2}$ and $v^V_{q_2}=\A(v_{q_2})$ respectively.

We will further use the notation $\A_{P_1}$ to denote the $Vf$-valued 1-form on $P_1$ induced by $\A$, i.e. the vertical part of the projection of a tangent vector to $P_1$ onto $Q_2$ via the tangent map of $\e^{(2)} \circ F$. Contraction of $\beta$ and $\A_{P_1}$ gives rise to a 1-form on $P_1$, namely \[ \langle\beta, \A_{P_1}\rangle (p_1) = \langle\beta(p_1),\A_{P_1}(p_1)\rangle \in T^*_{p_1}P_1\,. \]

Finally, we denote the $TQ_2$-component of the transformation $\psi_{L_2,\beta}: T_{P_1}Q_1 \to T_{P_2}Q_2$, i.e. the projection $\tau^{(2)}_1 \circ \psi_{L_2,\beta}$, by $\psi_{L_2,\beta}^{TQ_2}$. 

Then we can prove the following important result.

\begin{theorem}\label{thm:main} Let $(F,f)$ be a transformation pair between $\e^{(1)}:P_1\to Q_1$ and $\e^{(2)}:P_2\to Q_2$ and let $(\e^{(2)},L_2,\B_2)$ be a magnetic Lagrangian systems such that $L_2$ is $f$-regular. Fix a map $\beta$ as in Step~2 and let $\psi_{L_2,\beta}: T_{P_1}Q_1 \to T_{P_2}Q_2$ be the $(F,f)$-compatible diffeomorphism. Consider the magnetic Lagrangian system $(\e^{(1)},L_1,\B_1)$ defined by
\begin{enumerate}
\item $L_1(v_{q_1},p_1)=\left(\psi_{L_2,\beta}^*L_2\right)  (v_{q_1},p_1) -  \langle \beta
(p_1),\A(\psi_{L_2,\beta}^{TQ_2}(v_{q_1},p_1))\rangle$;
\item $\B_1= F^*\B_2 + d\left(\langle \beta , \A_{P_1}\rangle \right)$.
\end{enumerate}
Then $\psi_{L_2,\beta}$ is a symplectomorphism between the two symplectic structures associated with the two magnetic Lagrangian systems $(\e^{(1)}: P_1\to Q_1,L_1,\B_1)$ and $(\e^{(2)}: P_2\to Q_2,L_2,\B_2)$, and the corresponding Hamiltonian vector fields $X_{E_{L_1}}$ and $X_{E_{L_2}}$ are $\psi_{L_2,\beta}$-related.
\end{theorem}
\begin{proof}
The transformation pair $(F,f)$ induces a chain of bundle structures:
 \[
 P_1 \stackrel{F}{\to} P_2 \stackrel{\e^{(2)}}{\to} Q_2 \stackrel{f}{\to} Q_1.
 \]
As before, we choose coordinate charts that are adapted to these fibrations. The map $\psi_{L_2,\beta}$ has only nontrivial components in $\dot {\bar q}^a=\psi^a_{L_2,\beta}(q,\bar q, \dot q,\bar p,p)$. The map $\beta$ in coordinates reads $\beta(q,\bar q,\bar p,p) = \beta_a (q,\bar q,\bar p,p) d\bar q^a$ and $\Gamma^a_i(q,\bar q)$ denote the connection coefficients of $\A$, i.e.
 \[
 \A=(d\bar q^a + \Gamma^a_i dq^i) \otimes \fpd{}{\bar q^a} ,
 \]
and the vertical component of $v_{q_2}=(q^i,\bar q^a,\dot q^i,\dot{\bar q}^a)$ at $q_2=(q^i,\bar q^a)$ is $v^V_{q_2}=(q^i,\bar q^a,0,
\dot{\bar q}^a +\Gamma^a_j(q,\bar q)\dot q^j)$.

From the definition of $\psi_{L_2,\beta}$ we have the following identities:
\[
\frac{\partial L_2}{\partial \dot{\bar q}^a}(q,\bar q, \dot q, \psi_{L_2,\beta}(q,\dot q,\bar q,\bar p,p),\bar p) = \beta_a(q,\bar q,\bar p,p).
\]
The Lagrangian $L_1$ and the magnetic form $\B_1$ are then written as:
\begin{align*}
L_1(q,\bar q,\dot q, \bar p,p) &=L_2(q,\bar q,\dot q, \psi_{L_2,\beta}(q,\bar q,\dot q,\bar p,p), \bar p) - \beta_a(q,\bar q,\bar p,p) \big(\psi^a_{L_2,\beta}(q,\bar q,\dot q,\bar p,p) + \Gamma^a_i(q,\bar q)\dot q^i\big),\\
 \B_1 & =F^*\B_2 + d\bigg(\beta_a(q,\bar q,\bar p,p)(d\bar q^a + \Gamma_i^a(q,\bar q) dq^i) \bigg).
\end{align*}
The fact that $\psi_{L_2,\beta}$ is a symplectic diffeomorphism follows from a straightforward coordinate computation:
\begin{align*} \psi_{L_2,\beta}^*(\Omega^{L_2,\B_2}) &= \psi_{L_2,\beta}^*\left(d\left(\frac{\partial L_2}{\partial \dot q^ i}\right) \wedge dq^i +
d\left(\frac{\partial L_2}{\partial \dot{\bar q}^a}\right) \wedge d\bar q^a +\B_2\right) \\
&= d\left(\frac{\partial L_1}{\partial \dot q^ i}\right) \wedge dq^i + F^*\B_2 +d\bigg(\beta_a(q,\bar q,\bar p,p)(d\bar q^a + \Gamma_i^a(q,\bar q) dq^i) \bigg),
\end{align*}
i.e. $\psi_{L_2}^*(\Omega^{L_2,\B_2})=\Omega^{L_1,\B_1}$. It now remains to check that $\psi_{L_2,\beta}^*E_{L_2}=E_{L_1}$:
\begin{align*}
\psi_{L_2,\beta}^*E_{L_2} &= \psi_{L_2,\beta}^*\left(\dot q^i \frac{\partial L_2}{\partial \dot q^ i} + \dot{\bar q}^a \frac{\partial L_2}{\partial \dot{\bar q}^a}\right) - \psi_{L_2,\beta}^*L_2\\ &=
\dot q^i \frac{\partial L_1}{\partial \dot q^ i} -\left(\psi_{L_2,\beta}^*L_2 - \beta_a \Gamma^a_i \dot q^ i - \beta_a \psi_{L_2,\beta}^a \right).
\end{align*}
and the last term on the right-hand side is precisely $L_1$. To conclude, we remark that the two Hamiltonian vector fields $X_{E_{L_1}}$ and $X_{E_{L_2}}$ are $\psi_{L_2,\beta}$-related since $\psi_{L_2,\beta}$ is a diffeomorphism, $\psi_{L_2,\beta}^*(\Omega^{L_2,\B_2})=\Omega^{L_1,\B_1}$ and $\psi_{L_2,\beta}^*E_{L_2}=E_{L_1}$.  This completes the proof.
\end{proof}

\begin{remark}\textnormal{The fact that the two Hamiltonian vector fields are $\psi_{L_2,\beta}$-related, implies that every solution $p_1(t)\in P_1$ to the Euler-Lagrange equations for the system $(\e^{(1)},L_1,\B_1)$ projects under $F$ to a solution $p_2(t)=F(p_1(t))\in P_2$ of the Euler-Lagrange equations for the system $(\e^{(2)},L_2,\B_2)$.}\end{remark}

\begin{remark}\textnormal{The definition of the new Lagrangian $L_1$ and the magnetic 2-form $\B_1$ closely resembles the definition of the Routhian and magnetic 2-form in standard Routh reduction. In future work, we will show that the standard Routh reduction procedure itself can be regarded as the pull-back of the unreduced system with invariant Lagrangian $L$ under a transformation of the form $\psi_{L,\beta}$.}
\end{remark}

\section{Reduction by stages}\label{sec:red stages}
\subsection{Lagrangian systems on semi-direct products}
\paragraph{Generalities on semi-direct products.} Suppose we have a representation of a Lie group $G$ on a vector space $V$ and we write $gv$ for the action of $g \in G$ on $v \in V$. Likewise, the corresponding representation of the Lie algebra $\lag$ on $V$ is written as $\lag \times V \to V, (\xi,v) \mapsto \xi v$. The dual representation of $G$ on $V^*$ is defined by $G \times V^* \to V^*, (g,a) \mapsto g^*a$, where $\langle g^*a,v \rangle = \langle a ,gv \rangle$, for arbitrary $v \in V$. We then denote the semi-direct product of $G$ and $V$ by $GV=G\ltimes V$. The group multiplication of two elements $(g_1,v_1)$ and $(g_2,v_2)$ in $GV$ is
\[
(g_1,v_1)(g_2,v_2) = (g_1g_2,v_1+g_1v_2),
\]
the inverse of $(g,v)$ is $(g^{-1},-g^{-1}v)$. The Lie algebra $\lag V$ of $GV$ is the semi-direct product algebra $\lag\ltimes V$. The adjoint action of $GV$ on its Lie algebra is
\[
Ad_{(g,v)}(\xi,u) = (Ad_{g}\xi , gu - (Ad_g\xi)v),
\]
for $(g,v)\in GV$ and $(\xi,u)\in\lag V$ arbitrary. The Lie bracket of $(\xi_1,u_1)$ and $(\xi_2,u_2)$ in $\lag V$ equals
\[
[(\xi_1,u_1),(\xi_2,u_2)]_{\lag V} = \left([\xi_1,\xi_2]_\lag, \xi_1u_2-\xi_2u_1\right).
\]
The dual space of $\lag V$ is given by
\[
(\lag V)^* = \{(\mu, a)\,|\, \mu \in \lag^*, a \in V^*\}\,.
\]
The coadjoint action of $(g,v)\in GV$ on  $(\mu,a)\in (\lag V)^*$ is
 \[
Ad_{(g,v)}^*(\mu,a) = (Ad_{g}^*(\mu - v^*(a))  , g^*a)= (Ad_{g}^*\mu - (g^{-1}v)^*(g^*a), g^*a)\,,
\]
where  $v^*: V^*\to \lag^*$ is defined by $\langle v^*(a),\xi\rangle = \langle a,\xi v \rangle$.  It is a right action.

The closed subgroup $(e_G,V)$ of $GV$ is normal and $GV/(e_G,V)=G$. Similarly, $\lag V/(0,V) = \lag$. The isotropy subgroup of $(\mu,a)\in (\lag V)^*$ w.r.t. the coadjoint action is
\[
(GV)_{(\mu,a)} = \{ (g,v)\in GV\, | \,g^*a = a \mbox{ and } Ad_g^*(\mu -v^*(a))=\mu\}\,,
\]
from which it follows that $(g,v)\in (GV)_{(\mu,a)}$ in particular implies that $g \in G_a$, where $G_a$ is the isotropy subgroup of $a \in V^*$ under the action of $G$. In the specific case that the isotropy subgroup $G_a$ is trivial, i.e. $G_a = \{e_G\}$ with $e_G$ the unit element of $G$, any element in the isotropy group $(GV)_{(\mu,a)}$ is of the form $(e_G, v)$  with $v^*(a)=0$. In this case $(GV)_{(\mu,a)} $ determines a subgroup of the abelian group $V$. Throughout the following, we assume that $G_a=\{e_G\}$.

We now turn to a Lagrangian system $(Q,L)$ with configuration space $Q = S \times GV$ whereby the Lagrangian $L$ is supposed to be invariant under the (lifted) action of $GV$ onto the second factor. As we will see we can perform a Routh reduction with respect to the full semi-direct product group, or with respect to its abelian subgroup $V$. Both reduced systems are Lagrangian magnetic systems, and will be equivalent in the sense of Theorem~\ref{thm:main}.

\paragraph{$GV$-regularity of a $GV$-invariant Lagrangian $L$.} Following the definitions in the previous section, the Lagrangian $L$ determines a function $\ell: TS \times \lag V \to \R, (v_s, (\xi,u)) \mapsto \ell(v_s,\xi,u)$. Fix an element  $(\mu,a)\in (\lag V)^*$. The momentum relation $J_L\left(v_s,(g,v)(\xi,u)\right) = (\mu,a)$ is equivalent to the equations
\begin{align*}
&\F_2 \ell (v_s,\xi,u) = Ad_g^*(\mu-v^*a)\,,\\
&\F_3\ell(v_s,\xi,u)= g^*a\,,
\end{align*}
where $\F_2\ell$ and $\F_3\ell$ correspond to the fiber derivatives of $\ell$ with respect to the second and third argument, respectively.
The Lagrangian is $GV$-regular if for each $\left(v_s,(g,v)(\xi,u)\right) \in TS \times GV \times \lag V$ the map $\lag V \to (\lag V)^*, (\eta,w) \mapsto J_L\left(v_s,(g,v)(\xi + \eta, u + w)\right)$ is bijective. This translates into the existence, for each fixed $v_s \in TS$, of a mapping $(\chi_1,\chi_2): TS \times (\lag V)^* \to \lag V$ such that, for arbitrary $(v_s,(\nu,b)) \in TS \times (\lag V)^* $,
\begin{align*}
&\F_2 \ell \left(v_s, \chi_1(v_s, \nu,b),\chi_2(v_s, \nu,b)\right) = \nu\,,\\
&\F_3\ell\left(v_s, \chi_1(v_s,\nu,b),\chi_2(v_s,\nu,b)\right)= b\,,
\end{align*}

Throughout this section we assume moreover that a map $\tau: TS \times \lag \times V^* \to V$ exists such that
\[
\F_3\ell\left(v_s, \xi,\tau(v_s,\xi,b)\right)=b,
\]
for arbitrary $(v_s,\xi,b) \in TS \times \lag \times V^*$. From the $GV$-regularity it then follows that the condition $\F_2\ell\left(v_s, \xi, \tau(v_s,\xi,b)\right)=\nu$ is equivalent to $\xi=\chi_1(v_s,\nu,b)$ and, additionally, $\tau(v_s,\chi_1(v_s,\nu,b),b)=\chi_2(v_s,\nu,b)$.

We now show that the existence of the function $\tau$ is equivalent to $L$ being $V$-regular (where $V$ is identified with the abelian subgroup $\{e_G\} \times V$  of $GV$).  Recall that $L$ is $V$-regular if for every $(v_s,(g,v)(\xi,u))\in T(S \times GV) \cong TS \times GV \times \lag V$, and each $b\in V^*$, there is a unique $u'\in V$ such that
\[
\langle b,w \rangle =\left.\frac{d}{d\e} \right|_{\e=0} L\left(v_s, (g,v)(\xi,u +u' + \e g^{-1}w)\right) = \langle \F_3\ell(v_s,\xi,u+u'), g^{-1}w\rangle, \;\;\;\; \mbox{for all}\; w \in V.
\]
More specifically, for $(v_s,\xi) \in TS \times \lag$ arbitrary, there exists for every $b \in V^*$ a unique $u$ such that  $\F_3\ell(v_s,\xi,u)=b$. This precisely coincides with the existence of the map $\tau$.

\paragraph{Routh reduction w.r.t. $GV$.} Fix a local coordinate chart $(x^i)$ on $S$ and choose a regular momentum value $(\mu,a)\in (\lag V)^*$. We will write $\hat{\chi}_1(x,\dot{x},\cdot,)$ and $\hat{\chi}_2(x,\dot{x},\cdot)$ for the restrictions of the maps $\chi_1(x,\dot{x},\cdot)$ and $\chi_2(x,\dot{x},\cdot)$, respectively, to the coadjoint orbit $\Or_{(\mu,a)} \subset (\lag V)^*$ of $(\mu,a)$. According to what we found in Section~\ref{sec:3}, the reduced Euler-Lagrange equations of motion then become
\begin{align*}
& \dot \nu = ad^*_{\chi_1(x,\dot{x},\nu,b)}\nu -\left(\chi_2(x,\dot{x},\nu,b)\right)^*b, \\ & \dot b= \left(\chi_1(x,\dot{x},\nu,b)\right)^*b,\\
& \frac{d}{dt}\left( \fpd{\Ro^{(\mu,a)}_1}{\dot x^i}(x, \dot x, \nu,b)\right) - \fpd{\Ro^{(\mu,a)}_1}{x^i}(x, \dot x, \nu,b)=0,
\end{align*}
where the Routhian $\Ro^{(\mu,a)}_1$ is the function on $TS \times \Or_{(\mu,a)}$ given by
\[
\Ro^{(\mu,a)}_1(x,\dot{x},\nu,b)= \ell\left(v_s,\hat{\chi}_1(x,\dot{x},\nu,b),\hat{\chi}_2(x,\dot{x},\nu,b)\right) - \langle \nu , \hat{\chi}_1(x,\dot{x},\nu,b)\rangle  - \langle b, \hat{\chi}_2(x,\dot{x},\nu,b)\rangle\,.
\]

For later use we now compute the magnetic 2-form $\B_{(\mu,a)}$ explicitly. It is the trivial extension (i.e. the pull-back) to $S \times \Or_{(\mu,a)}$ of the reduction to $\Or_{(\mu,a)}\cong GV/GV_{(\mu,a)}$ of the exterior derivative of the 1-form $\A^1_{(\mu,a)}$ on $GV$ which is defined by
\begin{eqnarray*}
\A^1_{(\mu,a)} (g,v)\Big((g,v)(\xi,u)\Big) & = &  \langle \mu ,Ad_g\xi \rangle +\langle a, gu- (Ad_g\xi)v\rangle \\
& = & \langle Ad^*_g (\mu-v^*(a)), \xi\rangle + \langle g^*a,u\rangle.
\end{eqnarray*}
First, we introduce the following definition.
\begin{definition}\label{def:oneform}
$\theta_{(\mu,a)}$ is the 1-form on $\Or_{(\mu,a)}$ that satisfies
\[
\theta_{(\mu,a)}(\nu,b) (\dot \nu,\dot b) = \langle \nu, \xi\rangle,
\] with $(\dot \nu,\dot b = \xi^* b) \in T_{(\nu,b=g^*a)} \Or_{(\mu,a)} \subset (\lag V)^*$ arbitrary.
\end{definition}
This is well defined: by assumption, $G_{a}$ is trivial for all $a \in V^*$, hence there is a unique $\xi\in \lag$ such that $\dot b = \xi^* b$. Then we can prove:

\begin{lemma}\label{lemma:forceterms}
 $\B_{(\mu,a)}=d\theta_{(\mu,a)}.$
\end{lemma}
\begin{proof}
First note that the second term $\langle g^*a,u\rangle$ on the right-hand side of the above expression for $\A^1_{(\mu,a)}$ does not in fact contribute to the computation of $\B_{(\mu,a)}$: it is the contraction of the fixed `momentum' $a$ with the tangent vector $gu$ to the linear space $V$ and therefore vanishes when taking the exterior derivative.

Therefore, it is sufficient to show that $\theta_{(\mu,a)}$ is the reduction to $\Or_{(\mu,a)}$ of the 1-form $\A^1_{(\mu,a)}$ with the term $\langle g^*a,u\rangle$ omitted. For that purpose we write down the tangent map of the projection $GV \to GV/GV_{(\mu,a)} \cong \Or_{(\mu,a)}$:
\[
(g,v,g\xi, g u)\in T(GV) \mapsto (\nu = Ad^*_g(\mu-v^*(a)), b=g^*a, \dot \nu=ad^*_\xi\nu, \dot b =\xi^*b)\in T\Or_{(\mu,a)}.
\]
From this we can deduce that the first term on the right-hand side of the above expression for $ \A^1_{(\mu,a)}(g,v)\Big((g,v)(\xi,u)\Big)$ equals $\langle \nu,\xi\rangle$.
\end{proof}

\paragraph{Routh reduction w.r.t. $V$.} Here we use the $V$-principal connection on $S \times GV $
\[
\A^2(v_s,g\xi,gu) = gu\in V,
\]
which is the pull-back to $S \times GV$ of the standard $V$-principal connection on the Abelian group $V$. If $a$ is a regular value of the momentum map, the associated Routhian is a function on $T(S \times G)$ and equals $$\Ro^{a}_2(v_s, g\xi)= \ell(v_s, \xi,\tau(v_s,\xi,g^*a))- \langle g^*a, \tau(v_s,\xi,g^*a)\rangle.$$ In the present case the magnetic 2-form $\B_a$ vanishes.

\paragraph{The compatible transformation.}
 \[
\setlength{\arraycolsep}{.8cm}
\begin{array}{ccc}
\Rnode{a}{P_1= S \times \Or_{(\mu,a)}} & \Rnode{b}{ P_2 = S \times G}\\[1.5cm]
\Rnode{c}{Q_1=S} & \Rnode{d}{Q_2=S \times G}
\end{array}\psset{nodesep=5pt,arrows=->}
\everypsbox{\scriptstyle}
\ncLine{a}{b}\Aput{F}
\ncLine{a}{c}
\ncLine{d}{c}\Bput{f}
\ncLine{b}{d} \Aput{\mbox{id}}
\]

In the above diagram we introduce the different mappings involved. The map $F$ is determined from the projection $\Or_{(\mu,a)} \to G,\; (\nu, b) \mapsto g$ where $g$ is uniquely determined from $g^*a=b$. $f$ is simply the projection onto the first factor, and then $Vf= \ker Tf = 0_S\times TG  \subset T(S \times G)$. The pair $(F,f)$ is a transformation pair.
\begin{theorem}\label{thm:equivred} Assume that $G_a=\{e_G\}$ for $a \in V^*$ and that the map $\cdot^*a: V\to \lag^*$; $v\mapsto v^*a$ is onto. Then the two magnetic Lagrangian systems $(S \times \Or_{(\mu,a)}\to S, \Ro^{(\mu,a)}_1, \B_{(\mu,a)})$ and $(S \times G \to  S \times G, \Ro^{a}_2,0)$ are equivalent in the sense of Theorem~\ref{thm:main}, i.e. there is a $(F,f)$-compatible diffeomorphism of the form $\psi_{\Ro^{a}_2,\beta}$ and a connection $\A:TQ_2\to Vf$ such that the Lagrangians and the magnetic 2-forms satisfy
\begin{enumerate}
\item $\Ro^{(\mu,a)}_1(v_{q_1},p_1)=\left(\psi_{\Ro^a_2,\beta}^*\Ro^{a}_2\right)  (v_{q_1},p_1) -  \langle \beta (p_1),\A(\psi_{\Ro^a_2,\beta}^{TQ_2}(v_{q_1},p_1))\rangle$;
\item $\B_1= F^*\B_2 + d\left(\langle \beta , \A_{P_1}\rangle \right)$.
\end{enumerate}
The Hamiltonian vector fields $X_{E_{ \Ro^{(\mu,a)}_1}}$ and $X_{E_{ \Ro^{a}_2}}$ are $\psi_{\Ro^{a}_2,\beta}$-related. \end{theorem} 
\begin{proof} We now introduce the remaining elements needed to apply Theorem~\ref{thm:main}, i.e. a map $\beta$ and a connection $f:Q_2\to Q_1$:
\begin{enumerate}
\item The map $\beta$ is defined as follows: $\beta: P_1\to V^*f,\;(s,\nu,b)\mapsto (s,g, 0_s,\nu \circ TL_{g^{-1}})$ where $g\in G$ is such that $g^*a=b$. Note that the conditions $G_a=\{e_G\}$ and $\mbox{im} \cdot^*a = \lag^*$ imply that the fibres $F^{-1}(s,g) \cong V / \ker \cdot^*a$ and $V^*_{s}f\cong \lag^*$ are diffeomorphic. We show this by constructing pointwise an inverse for $\beta$. Given an arbitrary element in $V^*_{(s,g)}f$ and let $\nu$ be the corresponding element in the dual of the Lie algebra $\lag$.  Because $\cdot^*a$ is onto, a vector $v\in V$ exists such that $v^*a=\mu-Ad^*_{g^{-1}}\nu$. The element $(\nu,b=g^*a)$ then determines a point in $\Or_{(\mu,a)}$, is unique and by construction, it determines the inverse image for $\nu$ under $\beta|_{F^{-1}(s,g)}$.

\item The connection used to relate the dynamics is the pull-back to $S \times G$ of the standard zero-curvature connection with horizontal distribution $0_G\times TS\subset T(G\times S)$.
\end{enumerate}

Note that the contraction of the $\beta$-map and the vertical part of the standard connection precisely equals the 1-form $\theta_{(\mu,a)}$ on $\Or_{(\mu,a)}$: $(\nu,b, \dot \nu, \dot b)\mapsto \langle \nu, \xi \rangle$, with $\xi^* b = \dot b$. From Lemma~\ref{lemma:forceterms}, the exterior derivative of $\theta_{(\mu,a)}$ is precisely $\B_{(\mu,a)}$.

It now remains to show that the two Routhians $\Ro^{(\mu,a)}_1$ and $\Ro^{a}_2$ are transformed into each other by means of $\psi_{\Ro^a_2,\beta}$ and $\A$. For that purpose we derive an explicit formula for the second or `momentum' condition in Proposition~\ref{prop:compatible} (put $L_2=\Ro^a_2$). Let $(v_s,g\xi)$ be arbitrary in $T(S \times G)$.  Fix an element $g\eta\in T_gG$. Then
\begin{align*}
\langle \F_2 {\Ro^{a}_2}(v_s,g\xi), g\eta \rangle &= \left.\frac{d}{d\e}\right|_{\e=0} \Ro^{a}_2(v_s,g\xi+\e g\eta) \\ &=\left.\frac{d}{d\e}\right|_{\e=0} \Big(\ell(v_s,\xi+\e\eta,\tau(v_s,\xi+\e\eta,g^*a) )- \langle g^*a, \tau(v_s,\xi+\e\eta,g^*a)\rangle\Big)\\
&= \langle \F_2\ell (v_s,\xi,\tau(v_s,\xi,g^*a)), \eta\rangle.
\end{align*}
Therefore, to construct the transformation $\psi_{\Ro_2,\beta}$ we have to solve the following equation  for $\xi$:
\[
\F_2\ell (v_s,\xi,\tau(v_s,\xi,g^*a))=\beta(s,\nu,b)\circ TL_g = \nu.\]
By definition of $\tau$, the solution $\xi$ is precisely $\chi_1(v_s,\nu,b)$. From this, we necessarily have that the composition $\tau(v_s,\chi_1(v_s,\nu,b), b)$ equals $\chi_2(v_s,\nu,b)$. We now compute the transformation of $\Ro_2$ under $\psi_{\Ro_2,\beta}$ and $\A$:
\begin{align*}
&\Big(\Ro^{a}_2(v_s,g\xi) - \langle \nu, \xi \rangle\Big) \bigg|_{\F_2 \ell (v_s,\xi,\tau(v_s,\xi,g^*a)) = \nu} \\
& = \ell(v_s,\chi_1(v_s,\nu,b),\chi_2(v_s,\nu,b))- \langle b, \chi_2(v_s,\nu,b)\rangle- \langle \nu,\chi_1(v_s,\nu,b) \rangle.
\end{align*} This is precisely the Routhian $\Ro^{(\mu,a)}_1$ and using Theorem~\ref{thm:main}, this concludes the proof.
\end{proof}

\subsection{Example: Elroy's Beanie}
The example is taken from~\cite{marsdenphases}. The system consists of two planar rigid bodies that
are attached in their center of mass. The system moves in the plane
and is subject to some conservative force with potential $V$. The
configuration space is $S^1 \times SE(2)$, with coordinates $(\varphi, \theta, x,y)$. Here $(x,y)$ refers to the position of the center of
mass, $\theta$ is the rotation angle of the first rigid body, and $\varphi\in S^1$
denotes the relative rotation of the second body w.r.t.\ the first. The
kinetic energy of the system is $SE(2)$-invariant and we will also
suppose that the potential is $SE(2)$-invariant. This means in fact that only the relative position of the two bodies matters for the dynamics of the
system. The Lagrangian is of the form
\[
L=\frac12 m ({\dot x}^2 + {\dot y}^2) + \frac12 I_1 {\dot\theta}^2 +
\frac12 I_2 (\dot\theta+\dot\varphi)^2 -V(\varphi)\,.
\]
The Euler-Lagrange equations of the system are, written in normal form,
\[
\ddot x =0, \quad \ddot y=0, \quad \ddot\theta = \frac{1}{I_1} V',
\quad \ddot\varphi = - \frac{I_1+I_2}{I_1I_2} V'\,,
\]
where $V'= dV/d\varphi$.
\paragraph{The semi-direct product $SE(2)$.}
The special Euclidean group $SE(2)$ is the semi-direct product of the Lie group $G=S^1$ with $V=\R^2$, parametrized by $(\theta,x,y)$, where the action of $G$ on $V$ is defined by rotation in the plane. For convenience we identify $\R^2$ with $\C$ in the usual way: $(x,y)\mapsto z=x+iy$. Then the action of an element $\theta\in S^1$ on $z\in\C$ is by multiplication $e^{i\theta}z$.   

The identity of $SE(2)$ corresponds to $(\theta=0,z=0)$ and the group
multiplication is given by
\[
(\theta_1,z_1)*(\theta_2,z_2) =
( \theta_1+\theta_2,e^{i\theta_1}z_2 + z_1).
\]
Elements of the Lie-algebra $se(2)$ of $SE(2)$ are denoted by $(\xi,w)\in \R\times\C$. The associated infinitesimal action of the Lie algebra $\R$ of $S^1$ on $\C$ then reads $i\xi z$, with $\xi\in\R$ and $z\in \C$ arbitrary.  The adjoint action equals $Ad_{(\theta,z)} (\xi,w) = (\xi, e^{i\theta}w- i\xi z)$. If $(\theta,x,y,\dot\theta,\dot x,\dot y)$ is an element in $TSE(2)$, the corresponding element in the left identification with $SE(2)\times se(2)$ is $(\theta,z,\dot \theta,w)$, with $w= e^{-i\theta}\dot z$ and $\dot z= \dot x+i\dot y$. Denote the real and complex part of $w$ by $u,v$ respectively, $w=u+iv$.
This allows us to write down the Lagrangian $\ell$ on $TS^1 \times se(2)$ in the left identification as $$\ell(\varphi,\dot\varphi, \dot\theta, w) = \frac12 m (u^2 + v^2) + \frac12 I_1 {\dot\theta}^2 +
\frac12 I_2 (\dot\theta+\dot\varphi)^2 -V(\varphi).$$
Elements of the dual  $se^*(2) \cong \R\times \C$ of $se(2)$ are written as $(\mu, a)$ and the contraction with an arbitrary element $(\xi,w)\in se(2)$ is $\mu \xi + \mathfrak{Re} (a \overline w)$. The dual action $g^* a = e^{-i\theta} a$ and infinitesimally $\xi^* a= -i\xi a$. Clearly the isotropy group of $a$ is trivial for any $a$. Finally, for the element $z^*a$ in the dual of the Lie algebra of $S^1$ we obtain: $z^*a =  \mathfrak{Re}(-ia\overline z)$. The map $\cdot^*a: \C \to \R$ is onto for any $a\neq 0$.

\paragraph{Reduction with respect to $SE(2)$.}
The Lagrangian being of mechanical type we can compute the Routhian as follows:

\begin{align*} 2(\Ro^{(\mu,a)}_1+V)(\varphi,\dot\varphi,\nu,b) &= \left(\fpd{\ell}{\dot \varphi} \dot \varphi - \fpd{\ell}{\dot \theta} \dot \theta - \fpd{\ell}{u} u - \fpd{\ell}{v}
v\right)_{\scriptsize\left\{\begin{array}{l} \nu=(I_1+I_2)\dot \theta + I_2 \dot\varphi\\ b = m w\end{array} \right.}\\
&= \left(I_2 \dot \varphi^2 - (I_1+I_2)\dot \theta^2 - m u^2 - m
v^2\right)_{\scriptsize\left\{\begin{array}{l} \nu=(I_1+I_2)\dot \theta + I_2 \dot\varphi\\ b = m w\end{array} \right.}.\end{align*}
The momentum relations are regular:
$$\left\{\begin{array}{l} \dot \theta = \frac{\nu-I_2\dot\varphi}{I_1+I_2} = \chi_1(\varphi,\dot\varphi,\nu,b), \\ w = \frac{b}{m}=\chi_2(\varphi,\dot\varphi,\nu,b) . \end{array}\right.$$
Finally we obtain the Routhian after a straightforward computation:
\[
\Ro^{(\mu,a)}_1(\varphi,\dot\varphi,\nu,b)= \frac12 \frac{I_1I_2}{I_1+I_2} \dot \varphi^2 + \frac{I_2}{I_1+I_2} \nu\dot\varphi -V(\varphi) +  \frac{b\overline b}{2m} - \frac12\frac{\nu^2}{I_1+I_2}.
\]
\paragraph{The reduced equations of motion.}
The Routh reduced equations of motions then equal
\begin{align*}
&\dot \nu=\mathfrak{Re}\left(-\frac{i}{m} \overline b  b\right) =0, \\
&\dot b = -i \left(\frac{\nu-I_2\dot\varphi}{I_1+I_2}\right) b,\\
&\frac{d}{dt}\left( \fpd{\Ro^{(\mu,a)}_1}{\dot \varphi}(\nu,b,\varphi,\dot \varphi)\right) - \fpd{\Ro^{(\mu,a)}_1}{\varphi}(\nu,b,\varphi,\dot \varphi)=\frac{I_1I_2}{I_1+I_2} \ddot \varphi  + \frac{I_2}{I_1+I_2} \dot \nu + V'(\varphi)=0.
\end{align*}
The second equation of motion is clearly a rotation of the momentum $b$ with angular velocity $(I_2\dot\varphi-\nu)/(I_1+I_2)$. The choice of the fixed momentum $a$ is reflected in these equations as $b\overline b= a\overline a$.

\paragraph{Abelian reduction.}
We now perform Routh reduction w.r.t the abelian symmetry group $\R^2$ of translations in the $x$ and $y$ direction. Let us denote the symmetry group by $V=\R^2$ and study the quotient spaces. The conserved (complex) momentum for this action is $a=m(\dot x+i \dot y)$. We use the same momentum values as before: let $b = e^{-i\theta}a$. The map $\tau(\dot\theta, b, \varphi,\dot\varphi) =\frac{b}{m}$. The Routhian is obtained by computing
\begin{align*} 2(\Ro^{a}_2+V)(\theta,\varphi,\dot \theta,\dot\varphi) &= \left(\fpd{\ell}{\dot \varphi} \dot \varphi +\fpd{\ell}{\dot \theta} \dot \theta - \fpd{\ell}{u} u - \fpd{\ell}{v}
v\right)_{ e^{-i\theta}a= m w}\\
&= \left(I_1 \dot \theta^2 +I_2(\dot \theta+\dot\varphi)^2 - m w\overline w\right)_{e^{-i\theta}a= m w}\\ &=I_1 \dot \theta^2 +I_2(\dot \theta+\dot\varphi)^2 -\frac{a\overline a}{m}.  \end{align*}
Thus the Routh reduced system is a standard Lagrangian system on $S^1\times S^1$ with Lagrangian $ \frac12 I_1 {\dot\theta}^2 +  \frac12I_2 (\dot\theta+\dot\varphi)^2-V(\varphi)$ (we ignore irrelevant constant terms).

\paragraph{Equivalence.}
Using Theorem~\ref{thm:equivred}, both reduced systems are equivalent in the sense that a transformation $\psi_{\Ro_2^a,\beta}$  exists relating both magnetic Lagrangian systems. In this case, the $\beta$-map fixes the remaining momentum $\beta(\nu,b,\varphi)= \nu$. The diffeomorphism $\psi_{\Ro^a_2,\beta}$ then satisfies $\psi_{\Ro^a_2,\beta}(\varphi,\dot\varphi,\nu,b=e^{-i\theta} a) = (\theta,\varphi, \dot\varphi, \dot\theta= (\nu-I_2\dot\varphi)/(I_1+I_2))$. It then takes some tedious computations to see that $\Ro^{(\mu,a)}_1 = \psi_{\Ro^a_2,\beta}^* \Ro^a_2 - \nu (\nu-I_2\dot\varphi)/(I_1+I_2))$.
The 1-form $\theta_{(\mu,a)}$ on $\Or_{(\mu,a)}$ from Definition~\ref{def:oneform} equals $(\dot \nu,\dot b = -i \dot \theta b) \mapsto \nu \dot \theta$.

{\bf Acknowledgements.} BL is an honorary postdoctoral researcher at the Department of Mathematics of Ghent University and associate
academic staff at the Department of Mathematics of KU~Leuven. This work is sponsored by a Research Programme of the Research Foundation -- Flanders (FWO). This work is part of the {\sc irses} project {\sc
geomech} (nr.\ 246981) within the 7th European Community Framework Programme.

\bibliographystyle{plain}


\begin{thebibliography}{10}

\bibitem{adamec}
L.~{Adamec}.
\newblock {A route to Routh -- The classical setting}.
\newblock {\em J. Nonlinear Math. Phys.}, 18(1):87--107, 2011.

\bibitem{satellite}
A.~Bloch, P.S. Krishnaprasad, J.E. Marsden, and G.~S{\'a}nchez~De Alvarez.
\newblock Stabilization of rigid body dynamics by internal and external
  torques.
\newblock {\em Automatica}, 28:745--756, 1992.

\bibitem{CLpotential}
Anthony~M. Bloch, Dong~Eui Chang, Naomi~Ehrich Leonard, and Jerrold~E. Marsden.
\newblock Controlled {L}agrangians and the stabilization of mechanical systems.
  {II}. {P}otential shaping.
\newblock {\em IEEE Trans. Automat. Control}, 46(10):1556--1571, 2001.

\bibitem{CLmatching}
Anthony~M. Bloch, Naomi~Ehrich Leonard, and Jerrold~E. Marsden.
\newblock Controlled {L}agrangians and the stabilization of mechanical systems.
  {I}. {T}he first matching theorem.
\newblock {\em IEEE Trans. Automat. Control}, 45(12):2253--2270, 2000.

\bibitem{CLeuler}
Anthony~M. Bloch, Naomi~Ehrich Leonard, and Jerrold~E. Marsden.
\newblock Controlled {L}agrangians and the stabilization of
  {E}uler-{P}oincar\'e mechanical systems.
\newblock {\em Internat. J. Robust Nonlinear Control}, 11(3):191--214, 2001.

\bibitem{mestcram}
M.~{Crampin} and T.~{Mestdag}.
\newblock {Routh's procedure for non-Abelian symmetry groups}.
\newblock {\em {J}. {M}ath. {P}hys.}, 49(3):032901, 2008.

\bibitem{duistermaat2000lie}
J.J. Duistermaat and J.A.C. Kolk.
\newblock {\em Lie groups}.
\newblock Universitext (1979). Springer, 2000.

\bibitem{presympred}
A.~{Echeverr{\'{\i}}a-Enr{\'{\i}}quez}, M.~C. {Mu{\~n}oz-Lecanda}, and
  N.~{Rom{\'a}n-Roy}.
\newblock {Reduction of Presymplectic Manifolds with Symmetry}.
\newblock {\em Rev. Math. Phys.}, 11(10):1209--1247, 1999.

\bibitem{jalna}
S.M. Jalnapurkar and J.E. Marsden.
\newblock Reduction of {H}amilton's variational principle.
\newblock {\em Dynamics and Stability of Systems}, 15(3):287--318, 2000.

\bibitem{BC}
B.~Langerock, F.~Cantrijn, and J.~Vankerschaver.
\newblock Routhian reduction for quasi-invariant {L}agrangians.
\newblock {\em {J}. {M}ath. {P}hys.}, 51(2):022902, 2010.

\bibitem{routhstages}
B.~{Langerock}, T.~{Mestdag}, and J.~{Vankerschaver}.
\newblock {Routh reduction by stages}.
\newblock {\em SIGMA}, 7(109):31, 2011.

\bibitem{MarsdenHamRed}
J.E. Marsden, G.~Misio{\l}ek, J.P. Ortega, M.~Perlmutter, and Tudor~S. Ratiu.
\newblock {\em Hamiltonian reduction by stages}, volume 1913 of {\em Lecture
  Notes in Mathematics}.
\newblock Springer, Berlin, 2007.

\bibitem{marsdenphases}
J.E. Marsden, R.~Montgomery, and T.S. Ratiu.
\newblock {\em Reduction, Symmetry, and Phases in Mechanics}, volume~88 of {\em
  Mem. Amer. Math. Soc.}
\newblock 1990.

\bibitem{marsden1999introduction}
J.E. Marsden and T.S. Rațiu.
\newblock {\em Introduction to mechanics and symmetry: a basic exposition of
  classical mechanical systems}.
\newblock Texts in applied mathematics. Springer, 1999.

\bibitem{marsdenrouth}
J.E. Marsden, T.S. Ratiu, and J.~Scheurle.
\newblock Reduction theory and the {L}agrange-{R}outh equations.
\newblock {\em {J}. {M}ath. {P}hys.}, 41(6):3379--3429, 2000.

\bibitem{sympred}
J.E. Marsden and A.~Weinstein.
\newblock Reduction of symplectic manifolds with symmetry.
\newblock {\em Rep. Math. Phys.}, 5:121--130, 1974.

\bibitem{pars}
L.A. Pars.
\newblock {\em A Treatise on Analytical Dynamics}.
\newblock Heinemann Educational Books, 1965.

\bibitem{CLgyr}
Craig Woolsey, Chevva~Konda Reddy, Anthony~M. Bloch, Dong~Eui Chang,
  Naomi~Ehrich Leonard, and Jerrold~E. Marsden.
\newblock Controlled {L}agrangian systems with gyroscopic forcing and
  dissipation.
\newblock {\em Eur. J. Control}, 10(5):478--496, 2004.

\end{thebibliography}
\end{document}